\documentclass[orivec]{llncs}
\usepackage{color}
\usepackage{graphicx}
\usepackage{amssymb}
\usepackage{amsmath}
\usepackage{stmaryrd} %for \llbracket and \rrbracket
\usepackage{multicol}
\usepackage{hyperref}
\usepackage{url}
\usepackage{float}

\usepackage{import}
\subimport{}{ADT.sty}

\usepackage{pstricks}
\usepackage{mathtools}
\usepackage{pst-node,pst-3d,pst-tree,pst-all}
\usepackage{comment}
\usepackage{parskip}

\newcommand{\pp}[0]{{\mathrm{p}}}
\newcommand{\po}[0]{{\mathrm{o}}}

\newcommand{\counter}{{\mathrm{c}}}

\newcommand{\Out}{{\mathrm{Out}}}

\newcommand{\true}{{\mathit{true}}}
\newcommand{\false}{{\mathit{false}}}

\newcommand{\lf}{\mathrm{L}}
\newcommand{\nlf}{\mathrm{NL}}

\newcommand{\f}[1]{[#1]_{\mathrm{AD}}}
\newcommand{\g}[1]{\llbracket#1\rrbracket_{\mathrm{AD}}}
\newcommand{\fp}[1]{[#1]_{\mathrm{G}}}
\newcommand{\gp}[1]{\llbracket#1\rrbracket_{\mathrm{G}}}
\newcommand{\terms}{T_{\Sigma}}

\newcommand{\sem}[1]{\sat(#1)}
\DeclareMathOperator{\win}{Sat}
\DeclareMathOperator{\sat}{sat}
\DeclareMathOperator{\type}{type}
\DeclareMathOperator{\out}{out}

\newcommand{\winp}{\win^\pp_{\beta}}
\newcommand{\wino}{\win^\po_{\beta}}
\newcommand{\winpa}{\win^\pp_{\g{\sigma}}}

\newcommand{\winpp}{\win^\pp_{\beta^\pp}}
\newcommand{\winoo}{\win^\po_{\beta^\po}}
\newcommand{\winppa}{\win^\pp_{\g{\sigma^\pp}}}

\newcommand{\winpx}{\win^\pp}

\title{Attack--Defense Trees and Two-Player Binary Zero-Sum Extensive Form Games
Are Equivalent -- Technical Report with
Proofs\thanks{\href{http://dx.doi.org/10.1007/978-3-642-17197-0_17}{The original
publication is available at www.springerlink.com.}}}

\author{Barbara Kordy\thanks{B. Kordy was supported by the grant No. C08/IS/26 from FNR Luxembourg.}, Sjouke Mauw, Matthijs Melissen\thanks{M. Melissen was supported by the grant No. PHD--09--082 from FNR Luxembourg.}, Patrick Schweitzer\thanks{P. Schweitzer was supported by the grant No. PHD--09--167 from FNR Luxembourg.}}
\institute{University of Luxembourg}
\date{\today}

\begin{document}

\maketitle

\begin{abstract}
Attack--defense trees are used 
to describe  security
weaknesses of a system and possible countermeasures. 
In this paper, the connection between
attack--defense trees and game theory is made explicit. 
We show that attack--defense trees and binary zero-sum two-player
extensive form games have equivalent expressive power when considering 
satisfiability,
in the sense that they can be converted into each other while preserving their
outcome and  their internal structure.
\end{abstract}

\section{Introduction}

Attack trees~\cite{Schn2}, as popularized by Bruce Schneier at the end of the
1990s, form an informal but powerful method
to describe possible security weaknesses of a system. An attack tree
basically consists of a description of an attacker's goal and its
refinement into sub-goals. In case of a \emph{conjunctive} refinement,
all sub-goals have to be satisfied to satisfy the overall goal, while
for a \emph{disjunctive} refinement satisfying any of the sub-goals is
sufficient to satisfy the overall goal. The non-refined nodes (i.e.,
the leaves of the tree) are basic attack actions from which complex
attacks are composed.

Due to their intuitive nature, attack trees prove to be very useful in
understanding a system's weaknesses in an informal and interdisciplinary
context. The development of an attack tree for a specific system may
start by building a small tree that is obviously incomplete and
describes the attacks at a high level of abstraction, while allowing to
refine these attacks and to add new attacks later as to make a more
complete description.
Over the last few years, attack trees have developed into an
even more versatile tool. This is due to two developments. The first
development consists of the formalization of the attack trees
method~\cite{MaOo} which provides an attack tree with a precise meaning.
As a consequence, formal analysis techniques were
designed~\cite{WiJu,ReStFu} 
and computer tools were made commercially available~\cite{Program1,Program2}.

The second development comes from the insight that a more complete
description can be achieved by modeling the activities of a system's
defender in addition to those of the attacker.
Consequently, one can analyze which set of defenses is optimal from
the perspective of, for instance, cost effectiveness.
Several notions of protection trees or defense nodes have already been
proposed in the literature~\cite{EdDaRaMi,BiDaPe}. They mostly consist of adding
one layer of defenses to the attack tree, thus ignoring the fact that
in a dynamic system new attacks are mounted against these defenses and
that, consequently, yet more defenses are brought into place.
Such an alternating nature of attacks and defenses is captured in the
notion of attack--defense trees~\cite{KoMaRaSc}. 
In this recently developed extension of attack trees, 
the iterative structure of attacks and
defenses can be visualized and evolutionary aspects can be modeled.

These two developments, the formalization of attack trees and the
introduction of defenses, imply that an attack--defense tree
can be formally considered as a description of a game.
The purpose of this paper  is to make the connection between
attack--defense trees and game theory explicit.
We expect that the link between the relatively new field of attack
modeling and the well-developed field of game theory can be exploited
by making game theoretic analysis methods available to the attack
modeling community. As a first step, we study the relation between attack--defense trees
and games in terms of expressiveness. Rather than studying the
graphical attack--defense tree language, we  consider
an algebraic
representation of such trees, called \emph{attack--defense
terms} (ADTerms)~\cite{KoMaRaSc}, which allows for easier formal manipulation.

The main contribution  
of this paper is to show that ADTerms with a satisfiability attribute
are equivalent to two-player binary zero-sum extensive form games.
Whenever we talk about games, we 
refer to a game in this class. 
We show equivalence by defining two mappings: one from 
games to ADTerms and one from ADTerms to games. Then, we interpret a
\emph{strategy} in the game as a \emph{basic assignment} for the
corresponding ADTerm and vice versa.
Such a basic
assignment expresses which attacks and defenses are in place.
Equivalence then roughly means that for every winning strategy, there
exists a basic assignment that yields a satisfiable term, and vice versa.
Although the two formalisms have much in common, their equivalence
is not immediate. Two notions in the domain of
ADTerms have no direct correspondence
in the world of games: conjunctive nodes and refinements.
The mapping from ADTerms into games will have to solve this in a
semantically correct way.

This paper is structured as follows.
We introduce attack--defense terms and two-player binary zero-sum
extensive form games in Section~\ref{sec:prelim}. 
In Section~\ref{sec:gtoadt} we define a mapping from games to
attack--defense terms and prove that a player
can win the game if and only if he is successful in the corresponding ADTerm.
A reverse mapping is defined in Section~\ref{sec:adttog}.

This paper includes an appendix, containing the proofs that could not be included in
the main paper due to space restrictions.

\section{Preliminaries}
\label{sec:prelim}

\subsection{Attack--Defense Trees}

\label{sec:adt}
A limitation of attack trees is that they cannot capture the interaction between 
attacks carried out on a system and defenses put in place to fend off the attacks. 
To mitigate this problem and in order to be able to analyze an attack--defense scenario,
attack--defense trees are introduced in~\cite{KoMaRaSc}.
Attack--defense trees may have two types of nodes: attack nodes and defense nodes, representing 
actions of two opposing players.
The attacker and defender are modeled in a purely symmetric way. 
To avoid differentiating between attack--defense scenarios with an attack 
node as a root and a defense node as a root, the notions of \emph{proponent}
(denoted by $\pp$) and \emph{opponent} (denoted by $\po$) are introduced.
The root of an attack--defense tree represents the main goal of the proponent. To be more
precise, when the root is an attack node, the proponent is an attacker
and the opponent is a defender, and vice versa.

To formalize attack--defense trees we use attack--defense terms. 
Given a set $S$, we write $S^*$ for the set of all strings over $S$ and 
$\varepsilon$ for the empty string. 

\begin{definition}\label{def:ADTerms}
Attack--defense terms (ADTerms) are typed ground terms over a signature $\Sigma
=(S,F)$, where
\begin{itemize}
\item $S = \{\pp, \po \}$ is a set of types (we denote $-\pp=\po$ and $-\po=\pp$),
\item $F=\{(\vee^{\pp}_k)_{k\in\mathbb{N}},(\wedge^{\pp}_k)_{k\in\mathbb{N}},
(\vee^{\po}_k)_{k\in\mathbb{N}}, 
(\wedge^{\po}_k)_{k\in\mathbb{N}}, 
\counter^{\pp}, \counter^{\po}\}\cup \mathbb{B}^{\pp}\cup \mathbb{B}^{\po}$ 
is a set of functions equipped with a mapping 
$\type\colon F\to S^*\times S$, which expresses the type of each function
as follows.
For $k\in\mathbb{N}$,
\begin{align*}
&\type(\vee^{\pp}_k)=(\pp^k,\pp) &&\type(\vee^{\po}_k)=(\po^k,\po) \\
&\type(\wedge^{\pp}_k)=(\pp^k,\pp) &&\type(\wedge^{\po}_k)=(\po^k,\po)\\
&\type(\counter^{\pp})=(\pp\po,\pp) &&\type(\counter^{\po})=(\po\pp,\po)\\
&\type(b)=(\varepsilon,\pp), \text{ for\ } b\in\mathbb{B}^{\pp} && 
\type(b)=(\varepsilon,\po), \text{ for\ } b\in\mathbb{B}^{\po}.
\end{align*}
\end{itemize}
\end{definition}

The elements of $\mathbb{B}^{\pp}$ and $\mathbb{B}^{\po}$ 
are typed constants, which represent basic actions of the proponent and 
the opponent, respectively.
The functions $\vee^{\pp}_k,\wedge^{\pp}_k,\vee^{\po}_k, 
\wedge^{\po}_k$  
represent disjunctive ($\vee$) and conjunctive ($\wedge$)
refinement operators of arity $k$, for a proponent ($\pp$) and an opponent ($\po$), 
respectively. 
Whenever it is clear from the context, we omit the subscript $k$.
The binary function $\counter^s$ (`counter'), where $s\in S$, 
connects a term of the type $s$ with a countermeasure.
By $\terms$ we denote the set of all ADTerms. We partition $\terms$ into $\terms^\pp$ (the set
of terms of the proponent's type) and $\terms^\po$ (the set of terms of 
the opponent's type). To
denote the type of a term, we define a
function $\tau\colon \terms \to S$ by $\tau(t)=s$ if $t \in \terms^s$. 
\begin{example}\label{ex:ADTerm}
The ADTerm 
$t=\counter^{\pp}(\wedge^{\pp}(E,F),\vee^{\po}(G))\in \terms^\pp$
is graphically displayed in Fig.~\ref{fig:adt-and-game} (left).
For this ADTerm, we have
$\tau(t)=\pp$.
Subterms  $E$ and $F$ are basic actions of the proponent's type, 
and $G$ is a basic action of the opponent's type. Assuming the proponent
is the attacker, this means that the system can be attacked by combining
the basic attack actions $E$ and $F$. However the defender has the option
to defend if he implements the basic defense action $G$.
\end{example}

\begin{figure}[t]
\centering
\includegraphics{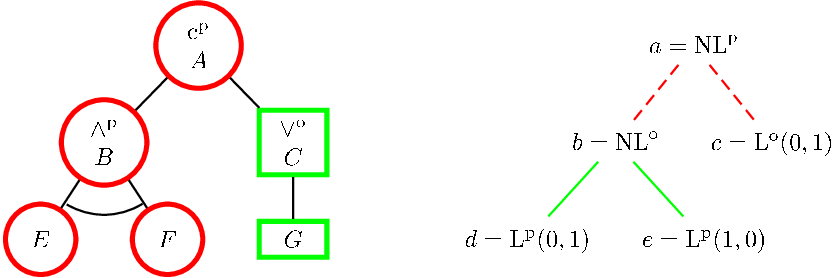}
	\caption{An example of an ADTerm (left) and a two-player binary zero-sum extensive form game (right).}
	\label{fig:adt-and-game}
\end{figure}

In order to check whether an attack--defense scenario is feasible, we
introduce the notion of satisfiability of an ADTerm by defining a satisfiability 
attribute $\sat$. First, for player $s\in\{\pp,\po\}$ we define a \emph{basic assignment} 
for $s$ as a function 
$\beta^s\colon \mathbb{B}^s \rightarrow \{\true, \false\}$.
We gather the basic assignments for both players in a \emph{basic assignment profile} 
$\beta = (\beta^\pp, \beta^\po)$. 
Second, the function 
$\sat\colon \terms \rightarrow \{\true, \false\}$ is used in order to calculate the 
satisfiability value of an ADTerm. It is defined recursively as follows 
\[
\sat(t)=
\begin{cases}
\beta^s(t^s), &\text{if $t=t^s\in \mathbb{B}^s$,} \\
\vee(\sat(t_1),\dots,\sat(t_k)), &\text{if $t=\vee^{s}(t_1,\dots,t_k)$},\\
\wedge(\sat(t_1),\dots,\sat(t_k)), &\text{if $t=\wedge^{s}(t_1,\dots,t_k)$},\\
\sat(t_1) \wedge \lnot \sat(t_2), &\text{if $t=\counter^s(t_1, t_2)$}.
\end{cases}
\]
For instance, 
consider the term $t$ from Example~\ref{ex:ADTerm} and the
 basic assignment profile $\beta=(\beta^{\pp},\beta^{\po})$, where
$\beta^{\pp}(E)=\true$, $\beta^{\pp}(F)=\true$, 
$\beta^{\po}(G)=\false$.
 We get $\sat(t)=\true$.
Assuming the proponent is the attacker, this means that the basic defense 
action $G$ is absent and the system is attacked by combining the basic
attack actions $E$ and $F$.

The next definition formalizes the notion of a satisfiable ADTerm for a player.

\begin{definition}\label{def:winning_sets} 
For every player $s$, strategy $\beta^s$ and strategy profile $\beta$, we define the sets of ADTerms $\win^s_\beta, \win^s_{\beta^s},  \win^s \subseteq \terms$ in the following way. Let $t \in \terms$.
\begin{itemize}
\item $t \in \win^s_\beta$ if either $\tau(t) = s$ and $\sem{t} = \true$, or $\tau(t) = -s$ and $\sem{t} = \false$. In this case we say that $s$ is successful in $t$ under $\beta$.
\item $t \in \win^s_{\beta^s}$ if $t \in \win^s_{(\beta^\pp, \beta^\po)}$ for every basic assignment $\beta^{-s}$. In this case we say that $s$ is successful in $t$ under $\beta^s$.
\item $t \in \win^s$ if there exists a basic assignment $\beta^s$ for player $s$ such that $t \in \win^s_{\beta^s}$. In this case we say that $t$ is satisfiable for $s$.
\end{itemize}
\end{definition}

\begin{theorem}\label{the:win}
For every ADTerm $t$, we have that every basic assignment profile $\beta$ partitions $\terms$ into $\winp$ and $\wino$.
\end{theorem}

\begin{proof}
This follows immediately from the first item in Definition~\ref{def:winning_sets}.
\end{proof}

\subsection{Two-player Binary Zero-sum Extensive Form Games}
We consider \emph{two-player binary zero-sum extensive form games},
in which a proponent $\pp$ and an opponent $\po$ play against each other.
In those games, we allow
only for the outcomes $(1,0)$ and $(0,1)$, where $(1,0)$ means
that the proponent succeeds in his goal (breaking the system if he is
the attacker, keeping the system secure if he is the defender), and
$(0,1)$ means that the opponent succeeds.
Note that the proponent is not necessarily the player who plays first in the game.
Finally, we
restrict ourselves to extensive form games, i.e., games in tree format.
Our presentation of games differs from the usual one, because
we present games as terms. This eases the transformation of games into ADTerms.
We formalize  games in the next definition, 
where  $\lf$ stands for a leaf and $\nlf$ for a non-leaf of the term.

\begin{definition}
Let $S = \{ \pp, \po \}$ denote the set of players and $\Out = \{ (1,0), (0,1)\}$
the set of possible outcomes. 
A two-player binary zero-sum extensive form game is a term
  $t ::= \psi^\pp\ |\ \psi^\po$, where
\begin{eqnarray*}
\psi^\pp & ::= & \nlf^\pp(\psi^\po, \ldots, \psi^\po)\ |\ \lf^\pp(1,0)\ |\ \lf^\pp(0,1) \\
\psi^\po & ::= & \nlf^\po(\psi^\pp, \ldots, \psi^\pp)\ |\ \lf^\po(1,0)\ |\ \lf^\po(0,1)  .
\end{eqnarray*}
We denote the set of all two-player binary zero-sum extensive form games
by $\mathcal{G}$. We define the \emph{first player} of a game $\psi^s$ as
the function $\tau\colon \mathcal{G} \rightarrow S$ such that $\tau(\psi^s) = s$.
\end{definition}

\begin{example}
An example of a two-player binary zero-sum extensive form game is the
expression $\nlf^\pp(\nlf^\po(\lf^\pp(0,1), \lf^\pp(1,0)), \lf^\po(0,1))$.
This game is displayed in Fig.~\ref{fig:adt-and-game} (right).
When displaying extensive form games, we use dashed edges for choices made by the proponent, 
and solid edges for those made by the opponent.
In this game, first the proponent can pick from two
options; if he chooses the first option, the opponent can choose
between outcomes $(0,1)$ and $(1,0)$. If the proponent chooses the
second option, the game will end with outcome $(0,1)$.
\end{example}

\begin{definition}
A function $\sigma^s$ is a \emph{strategy}
for a game $g \in \mathcal{G}$ for player $s\in S$
if it assigns to every non-leaf of player $s$ in $g$
$\nlf^s(\psi_1^{-s}, \ldots, \psi_n^{-s})$ a term 
$\psi^{-s}_k$ for some $k \in \{1, \ldots, n \}$.

A \emph{strategy profile} for a game $g \in \mathcal{G}$ is a pair $\sigma=(\sigma^\pp,
\sigma^\po)$, where $\sigma^\pp$ is a strategy of $g$ for $\pp$, and
$\sigma^\po$ a strategy of $g$ for $\po$.
\end{definition}

If $g = \nlf^s(\psi^{-s}_1, \ldots, \psi^{-s}_n)$ and
$\sigma = ( \sigma^\pp, \sigma^\po )$,
sometimes we abuse notation and write
$\sigma(g) = \psi_k^{-s}$ where $\psi_k^{-s} = \sigma^s(g)$.

Now we define the outcome of a game in three steps.

\begin{definition}
We say that $(0,1) \leq^\pp (1,0)$ and $(1,0) \leq^\po (0,1)$, so that
$(\Out, \leq^\pp)$ and $(\Out, \leq^\po)$ are totally ordered sets.
Let $(r^\pp, r^\po)$ be an element of $\Out$, and $\psi^{-s}_1, \ldots, \psi^{-s}_n$ be games with player $-s$ as the first player.

\begin{enumerate}
\item The \emph{outcome} $\out_{(\sigma^\pp, \sigma^\po)}\colon \mathcal{G}
\rightarrow \Out$ of a game $g$ under strategy profile $\sigma =
(\sigma^\pp, \sigma^\po)$ is defined by:
\begin{eqnarray*}
\out_{(\sigma^\pp, \sigma^\po)}(\lf^s(r^\pp, r^\po)) &=& (r^\pp, r^\po) \\
\out_{(\sigma^\pp, \sigma^\po)}(\nlf^s(\psi^{-s}_1, \ldots, \psi^{-s}_n)) &=& \out_{(\sigma^\pp, \sigma^\po)}(\sigma^s(\nlf^s(\psi^{-s}_1, \ldots, \psi^{-s}_n)))
\end{eqnarray*}
\item The \emph{outcome} $\out_{\sigma^s}\colon \mathcal{G} \rightarrow \Out$
of a game $g$ under strategy $\sigma^s$ is defined by:
\begin{eqnarray*}
\out_{\sigma^s}(\lf^{s}(r^\pp, r^\po)) &=& (r^\pp, r^\po) \\
\out_{\sigma^s}(\nlf^s(\psi^{-s}_1, \ldots, \psi^{-s}_n)) &=& \out_{\sigma^s}(\sigma^s(\nlf^s(\psi^{-s}_1, \ldots, \psi^{-s}_n))) \\
\out_{\sigma^s}(\nlf^{-s}(\psi^{-s}_1, \ldots, \psi^{-s}_n)) &=& \max\displaylimits_{1 \leq i \leq n}{}_{\leq^{-s}}\{\out_{\sigma^s}(\psi^{-s}_i) \}
\end{eqnarray*}
\item The \emph{outcome} $\out\colon \mathcal{G} \rightarrow \Out$ of a game
$g$ is defined by:
\begin{eqnarray*}
\out(\lf^s(r^\pp, r^\po)) &=& (r^\pp, r^\po) \\
\out(\nlf^{s}(\psi^{-s}_1, \ldots, \psi^{-s}_n)) &=& \max\displaylimits_{1 \leq i \leq n}{}_{\leq^s}\{\out_{\sigma^s}(\psi^{-s}_i) \} \\
\out(\nlf^{-s}(\psi^{-s}_1, \ldots, \psi^{-s}_n)) &=& \max\displaylimits_{1 \leq i \leq n}{}_{\leq^{-s}}\{\out_{\sigma^s}(\psi^{-s}_i) \}
\end{eqnarray*}
\end{enumerate}
\end{definition}

Here $\out_{(\sigma^\pp, \sigma^\po)}$ denotes the outcome of the game
when $\pp$ and $\po$ play according to strategy $\sigma^\pp$ and
$\sigma^\po$, respectively. Furthermore $\out_{\sigma^s}$ denotes the
outcome if player $s$ plays strategy $\sigma^s$, and player $-s$ tries
to achieve the best possible outcome for himself. Finally, $\out$ denotes
the outcome of the game if both players try to maximize their own
outcome. 

\section{From Games to ADTerms}
\label{sec:gtoadt}

In this section, we show how to transform binary zero-sum two-player
extensive form games into 
ADTerms. We define a function that transforms games into ADTerms, and
a function that transforms a strategy for a game into a basic
assignment for the corresponding ADTerm. First we show that the player
who wins the game is also the player for whom the
corresponding ADTerm is satisfiable, if both players play the basic assignment
corresponding to their strategy in the game. Then we show that if a
player has a strategy in a game which guarantees him to win, 
he is successful in the corresponding ADTerm under the corresponding basic assignment.
For this purpose, we first define a function $\f{\cdot}$ that maps games into ADTerms.

\begin{definition}
\label{def:f}
Let  $v^s$, $u^s$, and $u_1^s, \ldots, u_n^s$, for $s \in S$, represent
fresh basic actions
from $\mathbb{B}^s$. 
The function $\f{\cdot}\colon \mathcal{G} \rightarrow \terms$ is
defined in the following way.
\begin{subequations}
\begin{eqnarray}
\lf^\pp(1,0) & \mapsto & v^\pp \label{eqn:D6-1}\\
\lf^\po(1,0) & \mapsto & \counter^\po(u^\po, v^\pp) \label{eqn:D6-2}\\
\lf^\pp(0,1) & \mapsto & \counter^\pp(u^\pp, v^\po) \label{eqn:D6-3}\\
\lf^\po(0,1) & \mapsto & v^\po \label{eqn:D6-4}\\
\nlf^\pp(\psi_1,\ldots,\psi_n) & \mapsto & \vee^\pp(\counter^\pp(u_1^\pp,\f{\psi_1}), 
\ldots, \counter^\pp(u_n^\pp,\f{\psi_n})) \label{eqn:D6-5}\\
\nlf^\po(\psi_1,\ldots,\psi_n) & \mapsto & \vee^\po(\counter^\po(u_1^\po,\f{\psi_1}), 
\ldots, \counter^\po(u_n^\po,\f{\psi_n}))\label{eqn:D6-6}\text{.}
\end{eqnarray}
\end{subequations}
\end{definition}

The rules for player $\pp$ are visualized in Fig.~\ref{fig:gameadt}
(the rules for player $\po$ are symmetric).
The rules
specify that a winning leaf for a player in the
game is transformed into a satisfiable ADTerm for this player, i.e., an
ADTerm consisting of only a leaf belonging to this player (Rule~\eqref{eqn:D6-1}--\eqref{eqn:D6-4}),
and that non-leaves in the game are transformed into disjunctive
ADTerms of the same player (Rule~\eqref{eqn:D6-5}--\eqref{eqn:D6-6}). These disjunctions have children
of the form $\counter^s(u_k^\pp,\f{\psi_k})$ for some $k$. The intended
meaning here is that player $s$ selects $u_k^\pp$ exactly when his
strategy selects $\psi_k$ in the game.
An example of a transformation of a game into an ADTerm is depicted in
Fig.~\ref{fig:example-conversion}.

The resulting ADTerm is thus conjunction-free. Note that because terms in
games alternate between $\pp$ and $\po$, this procedure results in
valid ADTerms (i.e., in terms of the form $\counter^s(u^{s_1}, v^{s_2})$,
$s_1 = s$ and $s_2 = -s$, and disjunctive terms
for player $s$ have children for player $s$ as well).

\begin{figure}[t]
\centering
\includegraphics{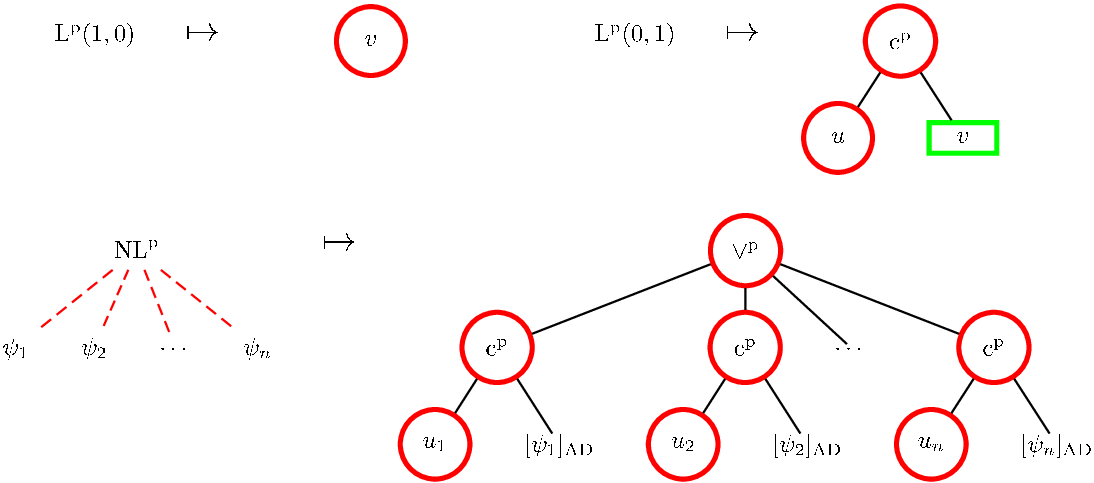}

\caption{Transformation of a game in extensive form into an ADTerm by function $\f{\cdot}$.}
\label{fig:gameadt}
\end{figure}

\begin{figure}[t]
\centering
\includegraphics{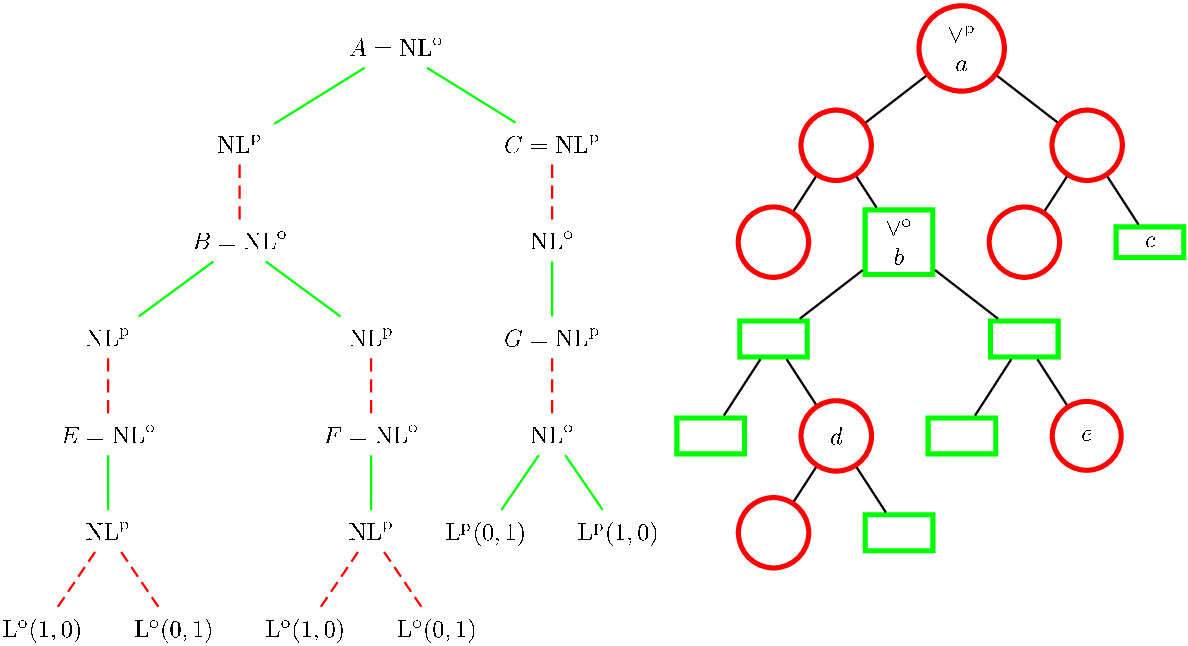}

	\caption{The result of the transformation of the ADTerm from Fig.~\ref{fig:adt-and-game} into a game (left), 
	and the game from Fig.~\ref{fig:adt-and-game} into an ADTerm (right).}
	\label{fig:example-conversion}
\end{figure}

Now we define how to transform a strategy profile for a game into a basic assignment profile for an ADTerm. First we define a transformation $\g{\cdot}$ from a strategy $\sigma^{s}$ ($s \in \{\pp,\po\}$) for game $g$ into a basic assignment $\beta^s = \g{\sigma^s}$ for ADTerm $\f{g}$. Intuitively, if a player's strategy for the game selects a certain branch, the basic assignment for the ADTerm assigns $\true$ to the node $u_k$ in the corresponding branch, and $\false$ to the nodes $u_k$ in the other branches. Furthermore, ADTerms resulting from leaves in the game are always selected.

\begin{definition}
Let $s$ be a player, $g$ be a game and $\sigma^s$ be a strategy of player $s$ for $g$. The function $\beta^s = \g{\sigma^s}$ is defined as follows. For all ADTerms $\counter^s(u^s,v^{-s})$ and $v^s$ resulting from the first four cases in Definition~\ref{def:f}, we set $\beta^s(u^s) = \beta^s(v^{s}) = \true$. For ADTerms obtained from game $g$ by one of the last two cases in Definition~\ref{def:f}, if $\sigma^s(g) = \psi_k$, we set $\beta^{s}(u_k^s) = \true$ and $\beta^{s}(u_i^s) = \false$ for $1 \leq i \leq n$, $i \neq k$.

The strategy profile $( \beta^\pp, \beta^\po )$ can be transformed into a basic assignment profile by $\g{ ( \beta^\pp, \beta^\po ) } = ( \g{ \beta^\pp}, \g{\beta^\po } )$.
\end{definition}

\bigskip

The next theorem states that a player is the winner in a game under a certain strategy profile if and only if he is successful in the corresponding ADTerm under the basic assignment profile corresponding to the strategy profile.

\begin{theorem}
\label{the:ga1}
Let $g$ be a game and $\sigma$ a strategy profile for $g$. Then
$\out_\sigma(g) = (1,0)$ if and only if $\f{g} \in \winpa$.
\end{theorem}

The following theorem states that a strategy in a game guarantees
player $s$ to win if and only if $s$ is successful in the corresponding ADTerm
under the corresponding basic assignment.
Surprisingly, this is not a consequence
of Theorem~\ref{the:ga1}: there might be a basic
assignment $\beta^s$ for the ADTerm, for which there exists no strategy
$\sigma^s$ such that $\beta^s = \g{\sigma^s}$
(i.e, the function $\g{\cdot}$ is not surjective).
Therefore it is not
immediately clear that if a player has a strategy $\sigma^s$ that wins from
the other player independent of his strategy, a player with a basic assignment
$\g{\sigma^s}$ wins from the other player independent of his basic assignment.

\begin{theorem}
\label{the:gsadttrat}
Let $g$ be a game and $\sigma^\pp$ be a strategy for $\pp$ on $g$. Then $\out_{\sigma^\pp}(g) = (1,0)$ if and only if $\f{g} \in \winppa$.
\end{theorem}

Now we obtain immediately the following corollary by definition of $\out$ and $\winpx$.

\begin{corollary}
\label{the:cor1}
Whenever $g$ is a game, $\out(g) = (1,0)$ if and only if $\f{g} \in \winpx$.
\end{corollary}

\section{From ADTerms to Games}
\label{sec:adttog}

We proceed with the transformation in the other direction. We define two transformations, namely from ADTerms into games, and from basic assignment profiles into strategy profiles. Then we show that if a player has a basic assignment for an ADTerm with which he is successful, the corresponding strategy in the corresponding game guarantees him to win.

\begin{definition}
\label{def:fadtgame}
We define a function $\fp{\cdot}$ from ADTerms to games as follows:
\begin{subequations}
\begin{eqnarray}
v^\pp & \mapsto & \nlf^\po(\nlf^\pp(\lf^\po(0,1), \lf^\po(1,0))) \label{def:fadtgame_a} \\
v^\po & \mapsto & \nlf^\pp(\nlf^\po(\lf^\pp(1,0), \lf^\pp(0,1))) \label{def:fadtgame_b} \\
\vee^\pp(\psi_1, \ldots, \psi_n) & \mapsto & \nlf^\po(\nlf^\pp(\fp{\psi_1},\ldots,\fp{\psi_n})) \label{def:fadtgame_c} \\
\vee^\po(\psi_1, \ldots, \psi_n) & \mapsto & \nlf^\pp(\nlf^\po(\fp{\psi_1},\ldots,\fp{\psi_n})) \label{def:fadtgame_d} \\
\wedge^\pp(\psi_1, \ldots, \psi_n) & \mapsto & \nlf^\po( \nlf^\pp(\fp{\psi_1}), \ldots, \nlf^\pp(\fp{\psi_n}) ) 
\label{def:fadtgame_e}\\
\wedge^\po(\psi_1, \ldots, \psi_n) & \mapsto & \nlf^\pp( \nlf^\po(\fp{\psi_1}), \ldots, \nlf^\po(\fp{\psi_n}) ) \label{def:fadtgame_f} \\
\counter^\pp(\psi_1, \psi_2) & \mapsto & \nlf^\po( \nlf^\pp(\fp{\psi_1}), \fp{\psi_2} ) 
\label{def:fadtgame_g} \\
\counter^\po(\psi_1, \psi_2) & \mapsto & \nlf^\pp( \nlf^\po(\fp{\psi_1}), \fp{\psi_2} )
\label{def:fadtgame_h}
\end{eqnarray}
\end{subequations}
\end{definition}

A graphical representation of the rules for player $\pp$ is displayed in Fig.~\ref{fig:adtgame}
(the rules for player $\po$ are symmetric).
It can easily be checked that this construction guarantees valid games
(in which $\pp$-moves and $\po$-moves alternate). According to these
rules, we transform leaves for player $s$ into two options for player
$s$, a losing and a winning one (Rules~\eqref{def:fadtgame_a} and \eqref{def:fadtgame_b}). These choices
correspond to not choosing and choosing the leaf in the
ADTerm, respectively. Disjunctive terms for player $s$ are transformed
into choices for player $s$ in the game (Rules~\eqref{def:fadtgame_c} and \eqref{def:fadtgame_d}). There is no
direct way of representing conjunctions in games. We can still handle
conjunctive terms though, by transforming them into choices for the
other player (Rules~\eqref{def:fadtgame_e} and \eqref{def:fadtgame_f}). This reflects the fact that a player
can succeed in all his options exactly when there is no way for the other player
to pick an option which allows him to succeed. Finally, countermeasures
against player $s$ are transformed into a choice for player $-s$ (Rules~\eqref{def:fadtgame_g} and \eqref{def:fadtgame_h}). Here, the first option corresponds to player
$-s$ not choosing the countermeasure, so that it is up to player $s$ whether
he succeeds or not, while the second option corresponds
to player $-s$ choosing the countermeasure.

The transformation of a game into an ADTerm is illustrated in
Fig.~\ref{fig:example-conversion}.

\begin{figure}[t]
\centering
\includegraphics{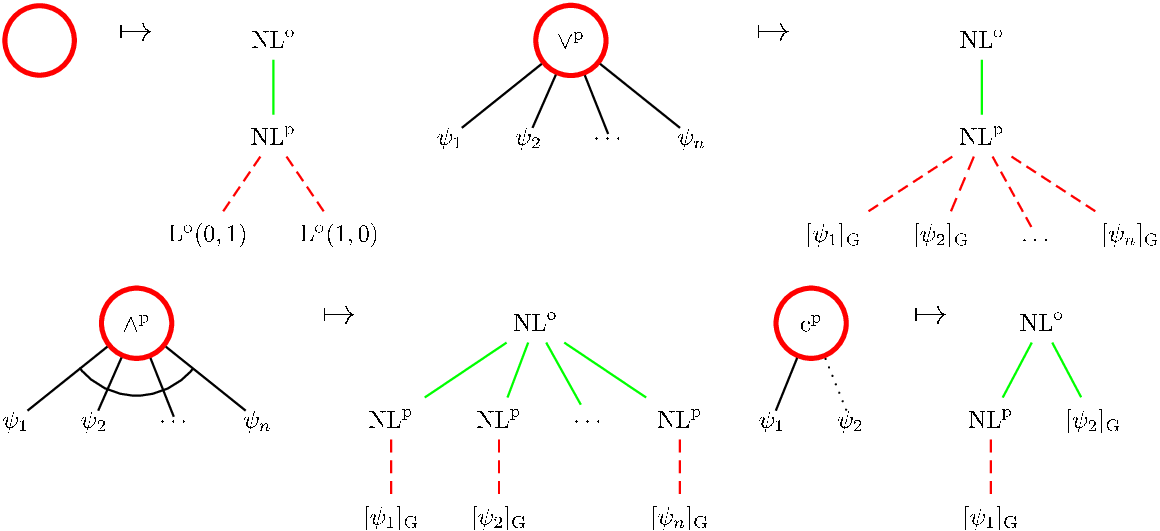}

\caption{Transformation of an ADTerm into a game by means of function
$\fp{\cdot}$.}
\label{fig:adtgame}
\end{figure}

We proceed by defining a transformation $\gp{\cdot}$ from a basic assignment for an ADTerm into a strategy for the corresponding game. We only give the definition for $s=\pp$; the definition for $s=\po$ is symmetric. 

\begin{definition}
\label{def:gp}
Function $\gp{\cdot}$ is a transformation from a basic assignment $\beta^\pp$ for ADTerm $t$ into a strategy $\sigma^\pp=\gp{\beta^\pp}$ for the game $\fp{t}$. If a (sub)term from $\fp{t}$ is obtained by rule $(2n)$ in Definition~\ref{def:fadtgame}, then $\sigma^\pp$ of that (sub)term is defined by rule $(3n)$ in this definition.

\begin{tabular}{lcllr}
	$\sigma^\pp(\nlf^\pp(\lf^\po(0,1), \lf^\po(1,0)))$&=&$\lf^\po(1,0)$ & if $\beta^s(v) = \true$. \hspace{3px} & $\mathrm{(3a)}$\\
	&=&$\lf^\po(0,1)$ & otherwise.\\
	$\sigma^\pp(\nlf^\pp(\fp{\psi_1},\ldots,\fp{\psi_n}))$&=&$\fp{\psi_k}$ & & $\mathrm{(3c)}$ \\
                         \multicolumn{4}{r}{where $k$ is the smallest number such that $\psi_k \in \winpp$.}\\
	&=&$\fp{\psi_1}$\\
                         \multicolumn{4}{r}{if there exists no such number.}\\
	$\sigma^\pp(\nlf^\pp( \nlf^\po(\fp{\psi_1}), \ldots, \nlf^\po(\fp{\psi_n}) ))$&=&$\nlf^\po(\fp{\psi_k})$&& $\mathrm{(3f)}$\\ 
                         \multicolumn{4}{r}{where $k$ is the smallest number such that $\psi_k \in \winpp$.}\\
	&=&$\nlf^\po(\fp{\psi_1})$\\
                         \multicolumn{4}{r}{if there exists no such number.}\\
	$\sigma^\pp(\nlf^\pp( \nlf^\po(\fp{\psi_1}), \fp{\psi_2} ))$&=&$\nlf^\po(\fp{\psi_1})$ & if $\psi_2 \not\in \winpp$.& $\mathrm{(3h)}$\\
	&=&$\fp{\psi_2}$ & otherwise.
\end{tabular}
For Rules (3b), (3d), (3e) and (3g), $\sigma^\pp$ is trivially defined as there is only one refinement.
\end{definition}

Note that some of the rules, namely Rules (3c), (3f) and (3h), are non-local in the sense that we need to evaluate all subterms of the ADTerm before we can decide what to play in the game.

\begin{theorem}
\label{the:adtgstrat}
Let $t$ be an ADTerm and $\beta^\pp$ a basic assignment for $t$. Then $t \in \winpp$ if and only if $\out_{\gp{\beta^\pp}}(\fp{t}) = (1,0)$.
\end{theorem}

Now the definition of $\out$ and $\winpx$ gives us the following corollary.

\begin{corollary}
\label{the:cor2}
Whenever $t$ is an ADTerm, $t \in \winpx$ if and only if $\out(\fp{t}) = (1,0)$.
\end{corollary}

\section{Conclusion}
We showed that attack--defense terms and binary zero-sum two-player
extensive form games have equivalent expressive power when considering satisfiability, 
in the sense
that they can be converted into each other while preserving their
outcome. 
Moreover, the transformations preserved internal structure, in the sense that
there exists injections between subterms in the game and subterms in the
ADTerm such that if a player wins in the subterm of the game, the
corresponding subterm in the ADerm is satisfiable for this player,
and vice versa. Therefore attack--defense trees with a satisfiability attribute and 
binary zero-sum
two-player extensive form games can be seen as two different
representations of the same concept.
Both representations have their
advantages. On the one hand, attack--defense trees are more intuitive,
because conjunctions and refinements can be explicitly modeled. On the other hand, the
game theory representation profits from the well-studied theoretical
properties of games.

We saw that two notions in the domain of
ADTerms had no direct correspondence to notions
in the world of games: conjunctive nodes and refinements.
The first problem has been solved by transforming conjunctive nodes for one player
to disjunctive nodes for the other player.
This also shows that, when considering the satisfiability attribute, 
the class of conjunction-free ADTerms has equal expressive power to the full class of ADTerms
(note that the transformation from ADTerms into games and vice versa are not
each other's inverse, i.e., $\fp{\f{t}} \neq t$ and $\f{\fp{t}} \neq t$). 
The second problem has been solved
by adding extra dummy moves with only one option for the other player in between refining
and refined nodes.

In the future, we plan to consider attack--defense trees accompanied with 
more sophisticated attributes, such that a larger class
of games can be converted. 
 An example of these are non-zero-sum games, where
$(1,1)$ can be interpreted as an outcome where both the attacker and
the defender profit (for example, if the attacker buys his goal from the
defender), and $(0,0)$ as an outcome where both parties are damaged 
(when the attacker fails in his goal, but his efforts damage the defender
in some way). Also the binary requirement can be lifted, so that
the outcome of a player represents for instance the cost or gain of 
his actions. 
Furthermore, it would be interesting to look for a correspondence of
incomplete and imperfect information in attack--defense trees.

\subsubsection*{Acknowledgments}
The authors would like to thank Leon van der Torre and Wojciech Jamroga for valuable discussions on the topic of this paper.

\bibliographystyle{splncs}
\bibliography{adt}

\appendix
\section{Appendix: proofs}

\begin{proof}[Theorem \ref{the:ga1}]
Let $g$ be a game and $\sigma$ a strategy profile for $g$. We set $t = \f{g}$ and $\beta = \g{\sigma}$. We prove the theorem by structural induction on term $g$. 

In the base case, $g=\lf^s(r^\pp, r^\po)$ with $r^\pp, r^\po \in \{0,1\}$.
If $\out_\sigma(g) = (1,0)$, $r^\pp = 1$ and $r^\po = 0$, so either $t=v^\pp$ or
$t=\counter^\po(u^\po, v^\pp)$. In the first case, $\sem{v^\pp} =
\true$ by definition of $g$, so $\sem{t} = \true$ and therefore $t \in
\winp$. In the second case, $\sem{u^\po} = \sem{v^\pp} = \true$ by
definition of $g$, so $\sem{t} = \sem{\counter^\po(u^\po, v^\pp)} =
\false$. Therefore $t \not\in \wino$ and thus $t \in \winp$. If
$\out_\sigma(g) \neq (1,0)$, $\out_\sigma(g) = (0,1)$, and by symmetry
$t \in \wino$ and thus $t \not\in \winp$.

We proceed with the induction step. We first give the proof that
$\out_\sigma(g) = (1,0)$ implies $t \in \winp$. We assume
$\out_\sigma(g) = (1,0)$. Then $g=\nlf^s(\psi_1,\ldots,\psi_n)$, so
$t=\vee^s(\counter^s(u_1^s,\f{\psi_1})), \ldots,
\counter^s(u_n^s,\f{\psi_n})$. Furthermore, there exists a $k$ such
that $\psi_k = \sigma(g)$. Therefore $\out_{\sigma}(\psi_k) = (1,0)$.
By induction hypothesis, $\f{\psi_k} \in \winp$. Now assume $s=\pp$,
so $\tau(\psi_k) = \po$ by the definition of a game. Then
$\tau(\f{\psi_k}) = \po$ by definition of $\f{\cdot}$, so
$\sem{\f{\psi_k}} = \false$ by Definition~\ref{def:winning_sets} and
$\sem{u_k^\pp} = \true$ by definition of $g$. Therefore
$\sem{\counter^\pp(u_k^\pp,\f{\psi_k})} = \true$, and therefore
$\sem{t} = \sem{\vee^s(\counter^s(u_1^s,\f{\psi_1}), \ldots,
\counter^s(u_n^s,\f{\psi_n})} = \true$. Because it holds that $\tau(t)
= \tau(g) = \pp$, we have $t \in \winp$.

Now assume $s = \po$, so $\tau(\psi_k) = \pp$ by the definition of a game. Then $\tau(\f{\psi_k}) = \pp$ by definition of $\f{\cdot}$, so $\sem{\f{\psi_k}} = \true$ by Definition~\ref{def:winning_sets}. Therefore it holds that $\sem{\counter^\pp(u_k^\pp,\f{\psi_k})} = \false$. Furthermore, whenever $1 \leq s \leq n$, $s \neq k$ $\sem{u^s} = \false$ by definition of $g$, and therefore $\sem{\counter^\pp(u^\pp,\f{\psi_k})} = \false$. This means that $\sem{t} = \sem{\vee^\po(\counter^\po(u_1^\po,\f{\psi_1}), \ldots, \counter^\po(u_n^\po,\f{\psi_n}))} = \false$. Because $\tau(t) = \tau(g) = \po$, $t \not\in \wino$ and therefore $t \in \winp$.

Now we show that $t \in \winp$ implies $\out_\sigma(g) = (1,0)$. We have that $\out_\sigma(g) = (0,1)$ implies $t \in \wino$ by symmetry. Then we also have $\out_\sigma(g) \neq (0,1)$ implies $t \not\in \winp$ by Theorem~\ref{the:win}. Then we have $t \in \winp$ implies $\out_\sigma(g) = (0,1)$ by contraposition.
\end{proof}

\begin{proof}[Theorem \ref{the:gsadttrat}]
Again we prove this by structural induction on the game. Let $t=\f{g}$ and $\beta^\pp = \g{\sigma^\pp}$ be a basic assignment.

The base case is the same as in Theorem~\ref{the:ga1}, so we just show the induction step. First we prove that $\out_{\sigma^\pp}(g) = (1,0)$ implies $\f{g} \in \winppa$. We assume $\out_{\sigma^\pp}(g) = (1,0)$. Then $g=\nlf^s((\psi_1,\ldots,\psi_n))$, so we have that $t=\vee^s(\counter^s(u_1^s,\f{\psi_1}), \ldots, \counter^s(u_n^s,\f{\psi_n}))$. Now assume $s=\pp$. Then there exists a $k$ such that $\psi_k = \sigma^\pp(g)$. Therefore $\out_{\psi_k} = (1,0)$. By induction hypothesis, $\f{\psi_k} \in \winpp$. Since we have $s=\pp$, $\tau(\psi_k) = \po$ by the definition of a game. Then $\tau(\f{\psi_k}) = \po$ by definition of $\f{\cdot}$, so $\sem{\f{\psi_k}} = \false$ by Definition~\ref{def:winning_sets} and $\sem{u_k^\pp} = \true$ by definition of $g$. Therefore $\sem{\counter^\pp(u_k^\pp,\f{\psi_k})} = \true$, and therefore $\sem{t} = \sem{\vee^s(\counter^s(u_1^s,\f{\psi_1}), \ldots, \counter^s(u_n^s,\f{\psi_n})} = \true$. Since we have $\tau(t) = \tau(g) = \pp$, $t \in \winpp$.

Now assume $s = \po$. Assume $1 \leq k \leq n$. Then $\out_{\psi_k} = (0,1)$, so $\f{\psi_k} \in \winpp$ by induction hypothesis. Furthermore, by $s = \po$,  $\tau(\psi_k) = \pp$ by the definition of a game. Then $\tau(\f{\psi_k}) = \pp$ by definition of $\f{\cdot}$, so $\sem{\f{\psi_k}} = \true$ by Definition~\ref{def:winning_sets}. Therefore $\sem{\counter^\pp(u_k^\pp,\f{\psi_k})} = \false$, so we have that $\sem{t} = \sem{\vee^\po(\counter^\po(u_1^\po,\f{\psi_1}), \ldots, \counter^\po(u_n^\po,\f{\psi_n}))} = \false$. Because $\tau(t) = \tau(g) = \po$, $t \not\in \winoo$ and therefore $t \in \winpp$.

We can prove in a similar way that $\out_{\sigma^\pp}(g) \neq (1,0)$ implies $\f{g} \not\in \winppa$, and from this we can prove that $\f{g} \in \win^\po_{\g{\beta_1^\po}}$ implies $\out_{\sigma^\po}(g) = (1,0)$ by contraposition.

\end{proof}

\begin{proof}[Corollary \ref{the:cor1}]
It holds that $\out(g) = (1,0)$ if and only if there is a strategy $\sigma^\pp$ for $\pp$ such that $\out_{\sigma^\pp} = (1,0)$. This holds if and only if $t \in \winpp$ for some $\beta^\pp$ by Theorem~\ref{the:gsadttrat}. This in turn is equivalent to $t \in \winpx$.
\end{proof}

\begin{proof}[Theorem \ref{the:adtgstrat}]
Let $g$ be $\fp{t}$. First assume that $t \in \winpp$. We prove by structural induction on the ADTerm that $\out_{\sigma^\pp}(g) = (1,0)$. We will treat the different cases according to the transformations in Definition~\ref{def:fadtgame}, which all have the form $t \mapsto g$; the first two cases are the base cases, the rest of the cases form the induction step. Note that because of $t \in\winpp$, $\tau(t) = \pp$ implies $\sem{t} = \true$ and $\tau(t) = \po$ implies $\sem{t} = \false$.

\begin{enumerate}
\item[(3a)] $v^\pp \mapsto \nlf^\po(\nlf^\pp(\lf^\po(0,1), \lf^\po(1,0)))$. Because $\sem{t} = \true$, $\beta^\pp(v) = \true$ as well, so $\sigma^\pp(\nlf^\pp(\lf^\po(1,0), \lf^\po(0,1)))=\lf^\po(1,0)$ by Definition~\ref{def:gp}. Therefore $\out_{\sigma^\pp}(g) = (1,0)$.
\item[(3b)] $v^\po \mapsto \nlf^\pp(\nlf^\po(\lf^\pp(1,0), \lf^\pp(0,1)))$. There exists a basic assignment $\beta^\po$ such that $t \not\in \win^\pp_{(\beta^\pp, \beta^\po)}$, namely when $\beta^\po(t) = \true$. Therefore $t \not\in \winpp$, which is a contradiction, so we don't need to consider this case.
\item[(3c)] $\vee^\pp(\psi_1, \ldots, \psi_n) \mapsto \nlf^\po(\nlf^\pp(\fp{\psi_1},\ldots,\fp{\psi_n}))$. Since $\sem{t} = \true$, there is a $j$ such that $\sem{\psi_j} = \true$. Let $k$ be the smallest number with this property. Then $\psi_{k} \in \winpp$, so $\sigma^\pp(\nlf^\pp(\fp{\psi_1},\ldots,\fp{\psi_n}))=\fp{\psi_{k}}$. Furthermore, by $\psi_{k} \in \winpp$, we have $\out_{\sigma^\pp}(\fp{\psi_{k}}) = (1,0)$ by induction hypothesis. Therefore $\out_{\sigma^\pp}(g) = (1,0)$. 
\item[(3d)] $\vee^\po(\psi_1, \ldots, \psi_n) \mapsto \nlf^\pp(\nlf^\po(\fp{\psi_1},\ldots,\fp{\psi_n}))$. By $t \in \winpp$ we have that $\sem{t} = \false$, so for $1 \leq k \leq n$, $\sem{\psi_k} = \false$. Furthermore, because $\tau(\psi_k) = \po$, it holds that $\psi_k \in \winpp$. By induction hypothesis, it holds that $\out_{\sigma^\pp}(\fp{\psi_k}) = (1,0)$, so $\out_{\sigma^\pp}(g) = \out_{\sigma^\pp}(\nlf^\pp(\nlf^\po(\fp{\psi_1},\ldots,\fp{\psi_n}))) = (1,0)$.
\item[(3e)] $\wedge^\pp(\psi_1, \ldots, \psi_n) \mapsto \nlf^\po( \nlf^\pp(\fp{\psi_1}), \ldots, \nlf^\pp(\fp{\psi_n}) )$. Assume $1 \leq k \leq n$. By $\sem{t} = \false$ we obtain $\sem{\psi_k} = \true$ and because $\tau(\psi_k) = \pp$, $\psi_k \in \winpp$. By induction hypothesis, $\out_{\sigma^\pp}(\fp{\psi_k}) = (1,0)$, so we have that $\out_{\sigma^\pp}(\nlf^\pp(\fp{\psi_k})) = (1,0)$ for $1 \leq k \leq n$, and therefore $\out_{\sigma^\pp}(g) = (1,0)$.
\item[(3f)] $\wedge^\po(\psi_1, \ldots, \psi_n) \mapsto \nlf^\pp( \nlf^\po(\fp{\psi_1}), \ldots, \nlf^\po(\fp{\psi_n}) )$. Because $\sem{t} = \false$, there is a number $j$ such that $\sem{\psi_j} = \false$. Let $k$ be the smallest number which this property. Then $\sigma^\pp(t)=\nlf^\po(\fp{\psi_k})$. Furthermore, by $\psi_{k} \in \winpp$, we have $\out_{\sigma^\pp}(\fp{\psi_{k}}) = (1,0)$ by induction hypothesis. Therefore $\out_{\sigma^\pp}(w_{k})) = (1,0)$ for $1 \leq k \leq n$, so $\out_{\sigma^\pp}(g) = (1,0)$.
\item[(3g)] $\counter^\pp(\psi_1, \psi_2) \mapsto \nlf^\po( \nlf^\pp(\fp{\psi_1}), \fp{\psi_2} )$. By $\sem{t} = \true$, we obtain $\sem{\psi_1} = \true$ and $\sem{\psi_2} = \false$. Because we have $\tau(\psi_1) = \pp$ and $\tau(\psi_2) = \po$, both $\psi_1 \in \winpp$ and $\psi_2 \in \winpp$. Then by induction hypothesis, $\out_{\sigma^\pp}(\fp{\psi_1}) = \out_{\sigma^\pp}(\fp{\psi_2}) = (1,0)$. Then $\out_{\sigma^\pp}(\nlf^\pp(\fp{\psi_1})) =size (1,0)$ as well, and therefore $\out_{\sigma^\pp}(g) = (1,0)$.
\item[(3h)] $\counter^\po(\psi_1, \psi_2) \mapsto \nlf^\pp( \nlf^\po(\fp{\psi_1}), \fp{\psi_2} )$. Because we have $\sem{t} = \false$, either $\sem{\psi_1} = \false$ or $\sem{\psi_2} = \true$. In the first case, $\sigma^\pp(t)=\fp{\psi_1}$. Moreover, by $\tau(\psi_1) = \po$, $\psi_1 \in \winpp$, so $\out_{\sigma^\pp}(\fp{\psi_1}) = (1,0)$ by induction hypothesis. Then $\out_{\sigma^\pp}(t) = (1,0)$ as well. In the second case, $\sigma^\pp(t)=\nlf^\po(\fp{\psi_1})$. Furthermore, by $\tau(\psi_2) = \pp$, we have $\psi_2 \in \winpp$, which implies $\out_{\sigma^\pp}(\fp{\psi_2}) = (1,0)$ by induction hypothesis. Then $\out_{\sigma^\pp}(g) = \out_{\sigma^\pp}(\nlf^\po(\fp{\psi_1})) = (1,0)$.
\end{enumerate}

In a similar way, we can prove that $t \not\in \winpp$ implies $\out_{\sigma^\pp}(g) \neq (1,0)$, which gives us by contraposition that $\out_{\sigma^\pp}(g) = (1,0)$ implies $t \in \winpp$.
\end{proof}

\begin{proof}[Corollary \ref{the:cor2}]
It holds that $t \in \winpx$ if and only if $t \in \winpp$ for some $\beta^\pp$. This holds if and only if there is a strategy $\sigma^\pp$ for $\pp$ such that $\out_{\sigma^\pp} = (1,0)$ by Theorem~\ref{the:adtgstrat}. This in turn is equivalent to $\out(g) = (1,0)$.
\end{proof}

\end{document}